\documentclass[runningheads]{llncs}
\usepackage{graphicx}
\usepackage{epstopdf}
\usepackage{multirow}
\usepackage{algorithm}
\usepackage{amssymb}
\usepackage[compatible]{algpseudocode}

\begin{document}

\title{Approximation Algorithms for Two-Bar Charts Packing Problem\thanks{The study was carried out within the framework of the state contract of the Sobolev Institute of Mathematics (project no. 0314--2019--0014).}}

\author{Adil Erzin\inst{1,2}\orcidID{0000-0002-2183-523X} \and Georgii Melidi\inst{2} \and Stepan
Nazarenko\inst{2} \and Roman Plotnikov\inst{1}\orcidID{0000-0003-2038-5609}}

\authorrunning{A. Erzin et al.}

\institute{Sobolev Institute of Mathematics, SB RAS, Novosibirsk 630090, Russia \and
Novosibirsk State University, Novosibirsk 630090, Russia\\
\email adilerzin@math.nsc.ru}

\maketitle              

\begin{abstract}
In the Two-Bar Charts Packing Problem (2-BCPP), it is required to pack the bar charts (BCs) consisting of two bars into the horizontal unit-height strip of minimal length. The bars may move vertically within the strip, but it is forbidden to change the order and separate the chart's bars. Recently, for this new problem, which is a generalization of the Bin Packing Problem (BPP), Strip Packing Problem (SPP), and 2-Dimensional Vector Packing Problem (2-DVPP), several approximation algorithms with guaranteed estimates were proposed. However, after a preliminary analysis of the solutions constructed by approximation algorithms, we discerned that the guaranteed estimates are inaccurate. This fact inspired us to conduct a numerical experiment in which the approximate solutions are compared to each other and with the optimal ones. To construct the optimal solutions or lower bounds for optimum, we use the Boolean Linear Programming (BLP) formulation of 2-BCPP proposed earlier and apply the CPLEX package. We also use a database of instances for BPP with known optimal solutions to construct the instances for the 2-BCPP with known minimal packing length. The results of the simulation make up the main content of this paper. 
\end{abstract}

\section{Introduction}
In \cite{Erzin20_1}, we studied the problem of optimizing an investment portfolio in the oil and gas field. Each project is characterized by the annual volume of hydrocarbon production, adequately displayed with bar charts. Each year, the total production volume of all projects must not exceed a given value (which may be a throughput of a pipe). The problem is to complete all projects as early as possible. This problem is a special case of the resource-constrained project scheduling problem with one renewable resource \cite{Brucker06}. In this case, the height of the BC's bar corresponds to the value of the consumed resource. Since the Bar Charts Packing Problem (BCPP) is intractable, we investigate a slight generalization of the Bin Packing Problem (BPP) when each BC consists of two bars. Let us denote BC consisting of $k$ bars as $k$-BC and the problem under consideration as 2-BCPP.

If this is not confusing, we will use the terms ``width'' and ``length'' interchangeably, implying horizontal dimensions. The 2-BCPP can be formulated as follows. We are given a set of 2-BCs. The height of each 2-BC's bar does not exceed 1. In a \emph{feasible} packing, the bars of each 2-BC do not change order and occupy adjacent cells; however, they can move vertically within the strip independently. In the 2-BCPP, it is required to find a feasible packing of the 2-BCs of minimal length. If we split the strip into equal unit-width cells of height 1, then the packing length is the number of cells in which there is at least one bar.

\subsection{Related results}
The 2-BCPP was first formulated in \cite{Erzin20_1} and then examined in \cite{Erzin20_2,Erzin20_3,Erzin20_4}, where the similar problems were described. These problems are the Bin Packing Problem (BPP) \cite{Johnson73}, the Strip Packing Problem (SPP)\cite{Baker80,Coffman80}, and the 2-D Vector Packing Problem (2-DVPP) \cite{Kellerer03}.

In the BPP, a set of items $L$ with given sizes is necessary to pack in the minimal number of unit-size bins. This problem is NP-hard. However, several approximation algorithms are known. One of them is the First Fit Decreasing (FFD). Items are numbered in non-increasing order, and the current item is placed in the first suitable bin.  It was proved that the FFD uses no more than $11/9\ OPT(L)+4$ bins \cite{Johnson73}, where $OPT(L)$ is the minimal number of bins to pack the items from $L$. Later the additive constant was reduced to 3 \cite{Baker85}, then it was reduced to 1 \cite{Yue91}, in 1997 to 7/9 \cite{Li97}, and finally, in 2007, the exact value equal to 6/9 of the additive constant was found \cite{Dosa07}. For the Modified First Fit Decreasing (MFFD) algorithm, it was shown that $MFFD(L)\leq 71/60\ OPT(L)+31/6$ \cite{Johnson85}. This estimate was improved to $71/60\ OPT(L)+1$ \cite{Yue95}.

In the SPP, it is necessary to pack (without rotation) a set of rectangles $R$ into the strip of minimal length. The Bottom-Left algorithm arranges rectangles in the descending order of height and yields a 3--approximate solution \cite{Baker80}. Then an algorithm with the ratio of 2.7  was proposed \cite{Coffman80}. In \cite{Sleator80} an algorithm that builds a 2.5--approximate solution was proposed. Later the ratio was reduced in \cite{Schiermeyer94}, and \cite{Steinberg97} to 2. The smallest known estimate for the ratio is $(5/3+\varepsilon)OPT(R)$, for any $\varepsilon > 0$ \cite{Harren14}.

2-DVPP is a generalization of BPP and a special case of 2-BCPP. It considers two attributes for each item and bin. The problem is to pack all items in the minimum number of bins, considering both attributes of the bin's capacity limits. In \cite{Kellerer03} a 2--approximation algorithm for 2-DVPP was presented. \cite{Christensen17} presents a survey of approximation algorithms for 2-DVPP. The best algorithm yields a $(3/2+\varepsilon)$--approximate solution, for any $\varepsilon >0$ \cite{Bansal16}.

In \cite{Erzin20_2}, we proposed an $O(n^2)$--time algorithm that builds a packing for $n$ 2-BCs, which length does not exceed $2OPT+1$, where $OPT$ is the minimum packing length for 2-BCPP. When at least one bar's height of each 2-BC is more than 1/2 (such 2-BCs we called ``big''), an $O(n^3)$--time 3/2--approximation algorithm was proposed. When each 2-BC is big and additionally non-increasing or non-decreasing, the complexity was reduced to $O(n^{2.5})$ preserving the ratio \cite{Erzin20_3}. The paper \cite{Erzin20_4} updates the estimates for the packing length of big 2-BCs, keeping time complexity. In \cite{Erzin20_4}, we improve the ratio and give a 5/4--approximation $O(n^{2.5})$--time algorithm for packing big non-increasing or non-decreasing 2-BCs. For the case of big 2-BCs (not necessarily non-increasing or non-decreasing), we proposed a $16/11$--approximation $O(n^3)$--time algorithm.

\subsection{Our contribution}
The main goal of this paper is a posteriori analysis of the previously developed algorithms. For this, we implement the approximation algorithms and conduct a simulation. To find an optimum or a lower bound for the packing length, we use the CPLEX package for the Boolean Linear Programming (BLP) problem. We also use a database of instances for BPP with known optimal solutions, from where we build instances for the 2-BCPP with known optimums.

\vspace{0.5cm}
The rest of the paper is organized as follows. Section 2 provides a statement of the 2-BCPP and the necessary definitions. In Section 3, we describe the algorithms under consideration and some properties. Section 4 is devoted to describing the numerical experiment results, and the last section concludes the paper.

\section{Formulation of the problem}
Let a semi-infinite unit-height horizontal strip be located on the plane so that its lower boundary coincides with the horizontal axis and its beginning is in origin. For each 2-BC $i$ from the given set $S$, $|S|=n$, consisting of two unit-width bars, the height of the first (left) bar is $a_i\in (0,1]$ and of the second (right) is $b_i\in (0,1]$. Let us split the strip into identical unit-width and unit-height rectangles, which we call the ``cells'', and number them with naturals starting from the beginning of the strip.

\begin{definition}
The \emph{packing} is a function $p:S\rightarrow\mathbb{Z}^+$ that assigns to each 2-BC $i$ a cell number $p(i)$ where its first bar falls. The packing is \emph{feasible} if the sum of the bar's heights that fall into each cell does not exceed 1.
\end{definition}

As a result of a packing $p$, the first bar of $i$th 2-BC falls into the cell $p(i)$ and the second bar falls into the cell $p(i)+1$. We will consider only feasible packings; therefore, the word ``feasible'' will be omitted.

\begin{definition}
The packing \emph{length} $L(p)$ is the number of strip cells in which at least one bar falls.
\end{definition}

In \cite{Erzin20_2}, we formulated 2-BCPP in the form of BLP. However, since in this paper we use the CPLEX package for the BLP to get optimal packings, we replicate its statement below for the reader's convenience. For this purpose, we introduce the variables:\\
$$
x_{ij}=\left\{
            \begin{array}{ll}
              1, & \hbox{if the first bar of the $i$th 2-BC is placed into the cell $j$;} \\
              0, & \hbox{else.}
            \end{array}
          \right.
$$
$$
y_j=\left\{
            \begin{array}{ll}
              1, & \hbox{if the cell $j$ contains at least one bar;} \\
              0, & \hbox{else.}
            \end{array}
          \right.
$$
Then 2-BCPP in the form of BLP is as follows.
\begin{equation}\label{e1}
  \sum\limits_j y_j \rightarrow\min\limits_{x_{ij},y_j\in\{0,1\}};
\end{equation}
\begin{equation}\label{e2}
  \sum\limits_j x_{ij} =1,\ i\in S;
\end{equation}
\begin{equation}\label{e3}
  \sum\limits_i a_ix_{ij} + \sum\limits_k b_kx_{k j-1}\leq y_j,\ \forall j.
\end{equation}

In this formulation, criterion (\ref{e1}) is the minimization of the packing length. Constraints (\ref{e2}) require each 2-BC to be packed into a strip once. Constraints (\ref{e3}) ensure that the sum of the bar's heights in each cell does not exceed 1.

The 2-BCPP (\ref{e1})-(\ref{e3}) is strongly NP-hard as a generalization of the BPP \cite{Johnson73}. Moreover, the problem is $(3/2-\varepsilon)$-inapproximable for any $\varepsilon >0$ unless P=NP \cite{Vazirani01}.

\section{Algorithms}
In \cite{Erzin20_2}, we had performed an analysis of some approximation algorithms. As a result of the simulation, we concluded that the greedy algorithm $GA\_LO$ with the preliminary lexicographic ordering of 2-BCs in a non-increasing manner significantly outperforms the other studied approximation algorithms. Therefore, in this article, for comparison with new algorithms, we use only one previously analyzed algorithm $GA\_LO$ \cite{Erzin20_2}.

New approximation algorithms were proposed in \cite{Erzin20_3} and \cite{Erzin20_4}. Below is a brief description of the algorithms we analyze in this paper, illustrated by an example.

\subsection{Greedy algorithm $GA\_LO$}
The algorithm $GA\_LO$ is described in detail (with pseudocode) in \cite{Erzin20_2}. Here is its short description. The algorithm sorts all 2-BCs lexicographically in non-increasing order of bar's height (Fig. \ref{fig_GA_LO}$a$). Let list $P$ be a lexicographically ordered set of 2-BCs. The first element in $P$ is placed in the first two strip's cells and removed from $P$. The 2-BCs deleted from $P$ never relocate further. Then until $P\neq\emptyset$, the following procedure is performed. For each 2-BC in $P$, search the leftmost position that does not violate the packing feasibility. Among 2-BCs that could be placed into the same leftmost cell, choose 2-BC with a minimal number, fix its position in the strip and remove it from $P$. For illustration, the resulting packing of the example in Fig. \ref{fig_GA_LO}$a$ is presented in Fig. \ref{fig_GA_LO}$b$.

\begin{figure}
\centering
\includegraphics[bb= 0 0 700 350, clip, scale=0.4]{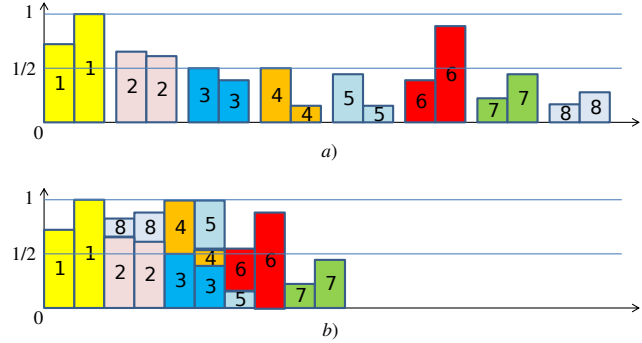}
\caption{$a$) Lexicographically ordered set of 2-BCs; $b$) Packing built by $GA\_LO$.} \label{fig_GA_LO}
\end{figure}

In \cite{Erzin20_2}, we proved that algorithm $A$, which uses $GA\_LO$ as a procedure, constructs a packing, whose length is at most $2OPT+1$, where $OPT$ is optimum. However, a posteriori analysis shows that algorithm $GA\_LO$ itself frequently constructs a near-optimal solution.

\subsection{Algorithms $M_w$ and $M1_w$}
Recall that for any natural $k$, we denote by $k$-BC a BC containing $k$ bars.

\begin{definition}
 Two arbitrary BCs create a $t$-union if they are packed so that $t$ cells of the strip contain the bars of both BCs.
\end{definition}

It follows from the definition that if two BCs consist of $x$ bars, their $t$-union (new BC) has $x-t$ bars. Thus, for example, in Fig. \ref{fig_GA_LO}$b$  the 2nd and the 8th 2-BCs form a 2-union, and the 5th and the 6th 2-BCs form a 1-union.

\begin{definition}
 If at least one bar in 2-BC is higher than 1/2, then such 2-BC we call \emph{big}. Consequently, we also call a bar \emph{big} if its height is more than 1/2. Otherwise, let us call a bar \emph{small}.
\end{definition}

The algorithm $M_w$ was described in detail in \cite{Erzin20_3}. It performs a sequence of steps. At each step, a max-weight matching is built in a specially constructed weighted graph. Initially, using the set $S$, we build a weighted graph $G_1=(V_1, E_1)$, in which the vertices are the images of 2-BCs in $S$ ($|V_1|=|S|=n$). The edge  $(i,j)\in E_1$ if the $i$th and $j$th 2-BCs can form a $t$-union ($t\in\{1,2\}$). The edge's weight equals 2 if the $i$th and $j$th 2-BCs can form a 2-union. If the $i$th and $j$th 2-BCs cannot form a 2-union, but can form a 1-union, then the weight of edge $(i,j)$ equals 1. Then in the graph $G_1$, a max-weight matching is constructed. The edges in the matching and their weights indicate the unions of 2-BCs. As a result of these unions, we get a new set of  2- and 3-BCs, which are the prototypes of vertices forming the set $V_2$ of the new weighted graph $G_2=(V_2, E_2)$. The edge $(i,j)$ exists in $G_2$ if the $i$th and $j$th BCs can form a union. If the $i$th and $j$th BCs can form a $t$-union with different $t>0$, we assign to $(i,j)$ the weight equal to the maximal $t$. At an arbitrary step in the corresponding graph $G_k$, we construct the next max-weight matching. The algorithm stops when in the next graph $G_{k+1}$, there are no edges.

\begin{theorem} \cite{Erzin20_3}
 The time complexity of algorithm $M_w$ is $O(n^4)$, and if all 2-BCs are big, it constructs a 3/2--approximate solution to the 2-BCPP.
\end{theorem}

From the proof of the theorem \cite{Erzin20_3} follows the

\begin{remark}
 In order to achieve the corresponding accuracy, it is sufficient to construct only the first max-weight matching. Thus, the complexity of obtaining a 3/2--approximate solution is $O(n^3)$.
\end{remark}

\begin{proposition}
  The algorithm $M_w$ builds only 1- and 2-unions.
\end{proposition}
\begin{proof}
Suppose that at some step, the algorithm $M_w$ builds a $t$-union with $t>2$ using some $p$-BC and $q$-BC, $p,q>2$. Initially, each BC consists of two bars, and only 1- and 2-unions are possible. Moreover, any 2-union of two 2-BCs forms a new 2-BC. Hence, at least one 1-union is needed to create later a BC with more than two bars. Therefore, the 1-unions appear during the formation of these $p$-BC and $q$-BC. Let us denote by $T$, $|T| = t$, the set of cells containing bars of both $p$-BC and $q$-BC. In the first cell of $T$, at least one bar of $p$-BC and one bar of $q$-BC are located. However, initially, each BC consists of two bars. Then at least one 2-BC of $p$-BC and one 2-BC of $q$-BC fall into the first and second cells of $T$. Otherwise, these $p$-BC and $q$-BC could not be obtained from the initial 2-BCs. Similarly, there exist two 2-BCs that fall into the $(t-1)$th and $t$th cells. This fact contradicts the $M_w$ performance since these two 2-unions should be formed before the 1-unions. Then at some step, there was a matching of non-maximum weight. Therefore, our assumption is wrong and $t$-unions with $t > 2$ are impossible. Fig. \ref{fig_proposition}$a$ shows the example of 3-union when $p,q=3$. In Fig. \ref{fig_proposition}$b$, the 1st and 3rd 2-BCs should have formed a 2-union at the previous step. The same is true for the 2nd and the 4th 2-BCs. The proposition is proved.
\end{proof}

\begin{figure}
\centering
\includegraphics[bb= 0 0 620 170, clip, scale=0.4]{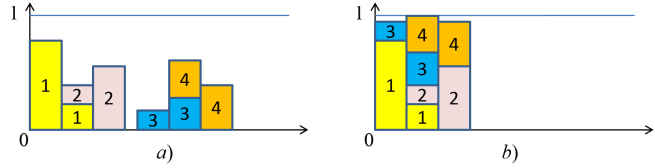}
\caption{$a$) $p$-BC and $q$-BC for 3-union when $p,q=3$; $b$) 3-union.} \label{fig_proposition}
\end{figure}

This proposition we use in the implementation of the algorithm $M_w$. That is, only 1- and 2-unions are built in this algorithm.

In the Simulation section, we also execute the algorithm $M1_w$, which constructs only the first max-weight matching in the graph $G_1$ (the first stage of algorithm $M_w$). Both algorithms $M1_w$ and $M_w$ have the ratio at most 3/2, but on average, these algorithms are more accurate in practice.

\begin{figure}
\centering
\includegraphics[bb= 0 40 700 350, clip, scale=0.4]{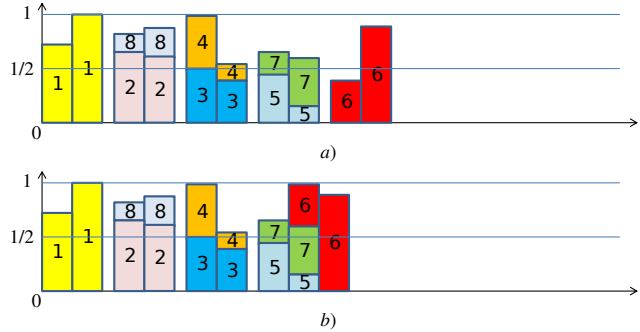}
\caption{$a$) Packing built by $M1_w$; $b$) Packing built by $M_w$.} \label{fig_M_w}
\end{figure}

Fig. \ref{fig_M_w}$a$ shows the first matching, which is the result of the algorithm $M1_w$. Fig. \ref{fig_M_w}$b$ shows the packing constructed by the algorithm $M_w$.

\subsection{Algorithms $A1$ and $A2$}
In \cite{Erzin20_4} an $O(n^{2.5}$)--time 5/4--approximation algorithm for packing non-increasing (or non-decreasing) big 2-BCs is proposed. The algorithm is based on the reduction of 2-BCPP to the Maximum Asymmetric Traveling Salesman Problem with edge's weights 0 or 1 (MaxATSP(0,1)) and using the algorithm proposed in \cite{Blaser04,Paluch18} for the latter problem. Furthermore, in \cite{Erzin20_4} an $O(n^3)$--time 16/11--approximation algorithm for the packing big (not necessary non-increasing or non-decreasing) 2-BCs is presented. To obtain this estimate, the algorithm for construction a max-cardinality matching \cite{Gabow83}, and approximation algorithm proposed in \cite{Paluch18} is used. We were interested in comparing this algorithm with the previously developed ones in two cases: all 2-BCs are big, and the 2-BCs are arbitrary. If 2-BCs are arbitrary, then we need to construct big 2-BCs to apply the considered algorithm. Therefore, we propose two different procedures for constructing big 2-BCs followed by applying the common part and call these algorithms $A1$ and $A2$.

\begin{figure}
\centering
\includegraphics[bb= 0 40 700 350, clip, scale=0.4]{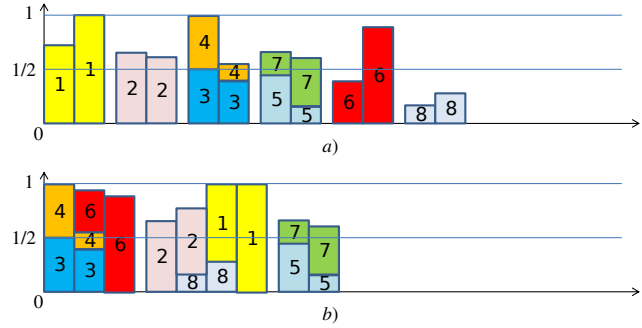}
\caption{$a$) Big BCs built by the first step of $A1$; $b$) Packing built by $A1$.} \label{fig_A1}
\end{figure}

In the algorithm $A1$, we use the first stage of algorithm $A$ described in \cite{Erzin20_2}. This procedure is as follows. Set $M = \emptyset$. The 2-BCs are browsing in numerical order. If the current 2-BC is big, then consider the next one. If both bars do not exceed 1/2 and $M = \emptyset$, then put the current 2-BC in $M$ and continue inspecting. If both bars do not exceed 1/2 and $M \neq \emptyset$, then form a 2-union of current 2-BC and 2-BC in $M$. If the resulting 2-BC is big, then exclude it from $M$ and consider the next 2-BC. If both bars of the resulting 2-BC do not exceed 1/2, leave it in $M$ and continue viewing. As a result of one scan of $S$, all 2-BCs, except possibly one, become big (Fig. \ref{fig_A1}$a$).

The first stage of the algorithm $A2$ consists of the sequential constructing of max-cardinality matchings in the graphs $G_k=(V_k,E_k)$, $k=1,2,\ldots$, where $V_k$ is the set of images of the current 2-BCs, and $(i,j)\in E_k$ if the $i$th and $j$th 2-BCs can form a 2-union (Fig. \ref{fig_A2}$a$). After each matching, we get new 2-BCs and similarly construct a new graph. This stage ends when there are no more 2-unions i.e., $G_k=(V_k,\emptyset)$. Thus, after the first step of the algorithm, there can be only 1-unions.

\begin{figure}
\centering
\includegraphics[bb= 0 40 700 350, clip, scale=0.4]{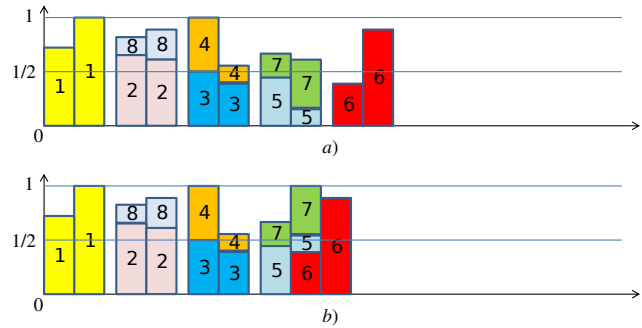}
\caption{$a$) Big BCs built by the first step of $A2$; $b$) Packing built by $A2$.} \label{fig_A2}
\end{figure}

A common part of both algorithms is identical. Using formed 2-BCs, the directed graph $G_1$ is constructed similarly as in the $M_w$. The vertices of this graph are the images of the big 2-BCs. The arc $(i, j)$ belongs to $G_1$ if the $i$th and $j$th 2-BCs can form a 1-union with  $i$th BC on the left. Then a so-called admissible multigraph is constructed \cite{Blaser04,Paluch18}, and its edges are colored with two colors such that the edges in each color class form a collection of node disjoint paths \cite{Blaser04(1)}. Then we choose the paths of one color with a maximal number of edges. This number of edges decreases the packing's length compared to the initial length, equal to $2m$, where $m\leq n$ is the number of big 2-BCs. Fig. \ref{fig_A1}$b$ shows the packing constructed by the algorithm $A1$, and Fig. \ref{fig_A2}$b$ shows the packing constructed by the algorithm $A2$.

\section{Simulation}
In this section, we describe a simulation that we performed on two groups of test data. The first one consists of randomly generated data. For the generation of the instances of the second group, we used the existing collection of the BPP instances with known optimal solutions. The considered algorithms have been implemented in the Python programming language. The calculations are carried out on the computer Intel Core i7-3770 3.40GHz 16Gb RAM.

\subsection{Randomly generated data}
The first group of instances is generated randomly with the different number of 2-BCs $n\in [25,1000]$. For each $n$, 50 different instances are generated. To build an optimal solution to BLP (1)-(3) or to find a lower bound for optimum, we use the IBM ILOG CPLEX 12.10 software package (CPLEX) with a limited running time of 300 seconds.

We examine the approximation algorithms $A1$, $A2$, $M1_w$, $M_w$ and $GA\_LO$ using three test data sets. The first one consists of arbitrary 2-BCs. For each 2-BC $i$, $i=1,\ldots,n$, the heights $a_i$ and $b_i$ take random uniformly distributed values in $(0,1]$. The second data set consists of big 2-BCs. While generating a 2-BC, one of its bars is chosen randomly as big, and its height takes a random uniformly distributed value in $(0.5, 1]$. Another bar's height takes a random uniformly distributed value in $(0, 1]$. The third data set consists of big non-increasing 2-BCs. Its generation is the same as in the second data set, but additionally, the bars of each generated 2-BC swap places if necessary.

\begin{figure}[t]
\includegraphics[width=\textwidth]{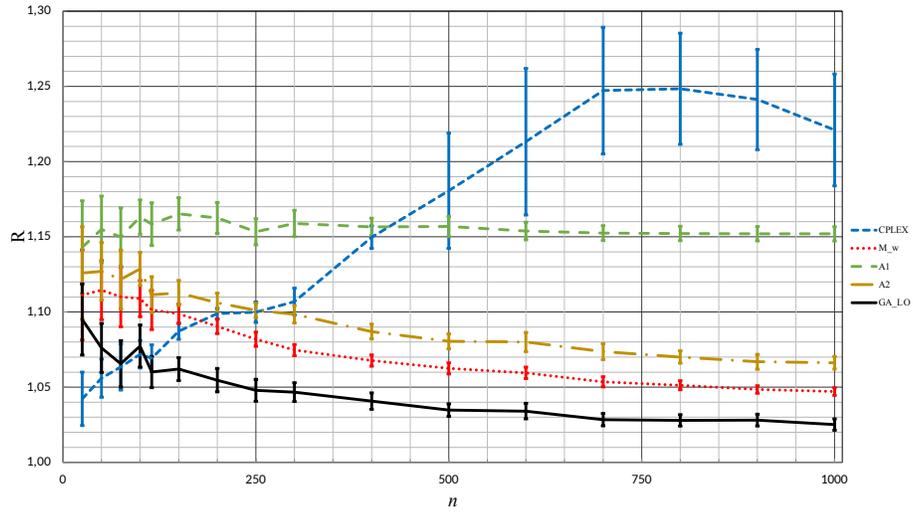}
\caption{Comparison of the algorithms on the randomly generated instances  with arbitrary 2-BCs.} \label{fig6}
\end{figure}

Fig. \ref{fig6} presents the numerical experiment results for the randomly generated instances with arbitrary 2-BCs. The value of $R$ is defined in the following way. If we know the optimum, then $R$ is the ratio $R = obj(X)/OPT$, where $OPT$ is the minimum packing length and $obj(X)$ is the length of the packing built by the algorithm $X\in\{A1, A2, M_w, GA\_LO\}$. On the other hand, if CPLEX fails to find the optimum during the allotted time, then we set $R = obj(X)/LB$, where $obj(X)$ is the objective's value of the approximate solution built by the algorithm $X \in \{$CPLEX$, A1, A2, M_w, GA$\_LO$\} $, and $LB$ is the lower bound of the optimum yielded by CPLEX during the allotted time. The figure shows the mean values and standard deviations (vertical segments) of $R$.

Even for small $n\leq 100$, CPLEX in 300 seconds often builds nonoptimal solution. However, the found by CPLEX approximate packing is near-optimal. In particular, when $n = 25$, CPLEX builds an optimal solution in 34\% of cases, and in 6\% of cases when $n \in \{50, 75\} $. However, when $n\geq 500$, CPLEX builds a significantly worse solution than other considered algorithms. Table \ref{table:1} shows minimum (min), maximum (max), and average (av) values of upper bounds on the \emph{absolute errors}, i.e., the differences between the algorithm's objective values and the lower bounds provided by CPLEX. The best values are marked in bold. For example, when $n = 1000$, CPLEX builds solutions with average absolute error 221.0, while the average absolute errors of the approximation algorithms $M_w$, $M1_w$, $A1$, $A2$, and $GA\_LO$ are 47.1, 47.4, 152.1, 66.3 and 25.2, correspondingly. As one sees in Table \ref{table:1} and Fig. \ref{fig6}, the algorithm $GA\_LO$  always builds the best solution among all considered algorithms. In most cases, algorithms $A2$ and $M_w$ turn out to be better than the algorithm $A1$. In about 95\% of cases, the first matching (the result of $M1_w$) turns out to be the only one in the algorithm $M_w$. It is important to note that all the tested approximation algorithms work fast enough for each $n\in [25,1000]$. The algorithm $GA\_LO$ appeared to be the fastest one. It always solves the problem in less than 1 second. Other approximation algorithms $M_w, A1$ and $A2$ need more running time, especially in the large-size cases. The average running time of these algorithms is about one minute when $n = 1000$.

Tables \ref{table:2} and \ref{table:3} present the simulation results for big 2-BCs and big non-increasing 2-BCs, respectively. Since all approximation algorithms yield almost the same values of $R$ in this case, we decide to include only algorithms $GA\_LO$ and CPLEX into the tables. Starting from $n = 250$, CPLEX builds a worse solution than all other considered algorithms. Therefore, we can conclude that on such a 2-BCs set, all algorithms work well, and we no longer observe such a difference in $R$ values as in the case with arbitrary 2-BCs (Fig. \ref{fig6}). For example, when $n \geq 250$, the average value of $R$ of all considered algorithms is about 1.2.

\begin{table}[h!]
\scriptsize{
\centering
\setlength{\tabcolsep}{0.6pt}
{
\begin{tabular}{|c|c|c|c|c|c|c|c|c|c|c|c|c|c|c|c|c|c|c|}
\hline
    \multirow{2}{*}{$n$} & \multicolumn{3}{|c|}{CPLEX} &  \multicolumn{3}{|c|}{$M_w$} & \multicolumn{3}{|c|}{$M1_w$} & \multicolumn{3}{|c|}{$A1$}
     & \multicolumn{3}{|c|}{$A2$} & \multicolumn{3}{|c|}{$GA\_LO$} \\
\cline{2-19}
 & $min$ & $max$& $av$ & $min$ & $max$& $av$ &$min$ & $max$& $av$ &$min$ & $max$& $av$ &$min$ & $max$& $av$ &$min$ & $max$& $av$ \\

\hline

25  &0 &3 &{\textbf{1.0}}	&0	&5	&2.8	&1	&6	&3.9	&0	&6	&3.6	&1	&7	&3.2	&0	&5   &2.4 \\
50  &0 &5 &{\textbf{2.8}}	&1	&9	&5.8	&3	&10	&6.5	&2	&11	&7.9	&2	&10	&6.4	&0	&7   &3.9 \\
75  &0 &9 &{\textbf{4.8}}	&1	&14	&8.3	&1	&14	&8.9	&5	&16	&11.4	&1	&15	&9.2	&0	&12  &5.0 \\
100  &1 &12 &{\textbf{7.2}}	&2	&15	&10.9	&2	&16	&11.3	&10	&23	&16.4	&9	&19	&12.9	&2	&15  &7.8 \\
115  &1 &11 &7.9	&2	&16	&11.6	&2	&18	&12.0	&11	&24	&18.2	&5	&19	&12.8	&2	&15  &{\textbf{6.9}} \\
150  &10 &18 &13.1 &12 &20 &14.8 &12 &20 &15.1 &18 &31 &24.7 &12 &25 &16.9 &5	&16  &{\textbf{9.3}} \\
200  &16 &25 &19.8 &14 &23 &18.1 &14 &23 &18.2 &24 &42 &32.5 &16 &27 &21.3 &5	&19  &{\textbf{11.0}} \\
250  &18 &39 &25.0 &16 &26 &20.5 &16 &27 &20.8 &27 &47 &38.4 &19 &34 &25.3 &5	&20  &{\textbf{12.1}} \\
300  &23 &51 &32.2 &18 &27 &22.5 &18 &27 &22.5 &37 &61 &47.8 &23 &38 &29.7 &8	&24  &{\textbf{14.1}} \\
400  &48 &87 &60.0 &21 &33 &27.1 &21 &33 &27.3 &52 &71 &62.7 &28 &50 &34.9 &8	&29  &{\textbf{16.4}} \\
500  &61 &188 &89.9 &24 &39 &31.1 &24 &39 &31.1 &60 &95 &78.1 &30 &54 &40.1 &9 &31  &{\textbf{17.3}} \\
600  &70 &255 &128.1 &28 &50 &35.7 &29 &50 &36.0 &76 &108 &92.3 &32 &72 &48.1 &11 &38  &{\textbf{20.5}} \\
700  &78 &258 &173.1 &27 &49 &37.5 &28 &49 &37.7 &93 &120 &106.8 &37 &66 &51.6 &12 &28  &{\textbf{19.9}} \\
800  &92 &287 &198.6 &34 &56 &41.0 &34 &56 &41.1 &102 &136 &121.7 &37 &73 &56.0 &11 &39  &{\textbf{22.4}} \\
900  &105 &335 &217.4 &34 &55 &43.7 &34 &56 &43.9 &119 &159 &137.0 &41 &88 &60.4 &14 &54  &{\textbf{25.4}} \\
1000  &112 &343 &221.0 &39 &64 &47.1 &39 &64 &47.4 &129 &175 &152.1 &53 &93 &66.3 &13 &45  &{\textbf{25.2}} \\

\hline

\hline
\end{tabular}
\caption{Absolute errors of the algorithms on the randomly generated instances with arbitrary 2-BCs.}
\label{table:1}
}}
\end{table}

\begin{table}[h]
\centering
\begin{tabular}{|c|c|c|c|c|}
\hline
\multirow{2}{*}{$n$} & \multicolumn{2}{c|}{CPLEX}     & \multicolumn{2}{c|}{\textit{$GA\_LO$}} \\ \cline{2-5}
        & \textit{Rav}    & \textit{Rsd} & \textit{Rav}       & \textit{Rsd}    \\ \hline
25      & \textbf{1.003}  & 0.011        & 1.013             & 0.02          \\ \hline
50      & \textbf{1.003} & 0.008       & 1.009           & 0.01          \\ \hline
75      & \textbf{1.004} & 0.007       & 1.008             & 0.008          \\ \hline
100     & \textbf{1.003} & 0.004       & 1.006             & 0.005          \\ \hline
250     & 1.207           & 0.014       & \textbf{1.203}    & 0.015          \\ \hline
500     & 1.242          & 0.025       & \textbf{1.203}    & 0.01          \\ \hline
750     & 1.237          & 0.01       & \textbf{1.201}    & 0.01          \\ \hline
1000    & 1.234          & 0.01       & \textbf{1.202}    & 0.007          \\ \hline
\end{tabular}
\caption{Comparison of the algorithms on the randomly generated instances when all 2-BCs are big ($R_{av}$ is the mean value; $R_{sd}$ is the standard deviation).}
\label{table:2}
\end{table}

\begin{table}[h]
\centering
\begin{tabular}{|c|c|c|c|c|}
\hline
\multirow{2}{*}{$n$} & \multicolumn{2}{c|}{CPLEX}     & \multicolumn{2}{c|}{\textit{$GA\_LO$}} \\ \cline{2-5}
        & \textit{Rav}    & \textit{Rsd} & \textit{Rav}       & \textit{Rsd}    \\ \hline
25      & \textbf{1.005}  & 0.013        & 1.022              & 0.017          \\ \hline
50      & \textbf{1.017}  & 0.028       & 1.025             & 0.025          \\ \hline
75      & \textbf{1.018} & 0.02       & 1.025             & 0.02          \\ \hline
100     & \textbf{1.022} & 0.043       & 1.026             & 0.042          \\ \hline
250     & 1.235          & 0.018       & \textbf{1.216}     & 0.014          \\ \hline
500     & 1.26          & 0.021       & \textbf{1.21}    & 0.009          \\ \hline
750     & 1.247          & 0.012       & \textbf{1.208}    & 0.009           \\ \hline
1000    & 1.247          & 0.008       & \textbf{1.204}    & 0.009          \\ \hline
\end{tabular}
\caption{Comparison of the algorithms on the randomly generated instances when all 2-BCs are big and non-increasing ($R_{av}$ is the mean value; $R_{sd}$ is the standard deviation).}
\label{table:3}
\end{table}

Subsequently, we conclude that $GA\_LO$ is preferable among all tested algorithms for the input data with random uniformly distributed parameters. Starting from $n = 100$, it builds more accurate solutions than other algorithms. However, considering the particular cases of the 2-BCs (second and third data sets), it is hard to single out the most efficient algorithm. In these cases, the only definite advantage of the algorithm $GA\_LO$ is its running time, which for all $n$ does not exceed 1 second. For instance, when $n=1000$, the algorithm $GA\_LO$ builds solutions in 0.37 seconds.

As a result of this part of the simulation, we conclude that the algorithm $A2$ turns out to be better than $A1$.  However, the algorithm $M_w$ solves the problem a little accurately than $A2$. The difference between $R$ values of $M_w$ and $A2$ is about 0.02 for $n \in \{700, 800, 900, 1000\} $. The difference between $R$ values of $A2$ and $GA\_LO$ is about 0.04 for the same $n$. The running time of the algorithms $A2$ and $M_w$ is almost the same, and on average, for $n = 1000$, they build solutions in 68 and 61 seconds, respectively.

To make our results reproducible, we have uploaded all the instances to the cloud storage publicly available at the link:\\ https://disk.yandex.ru/d/sb1IqReFpUu7dg.

\subsection{Data generated from the existing BPP benchmarks}
Another group of test cases in our simulation is based on the Bin Packing Problem (BPP) instances with known optimal solutions. In \cite{Delorme} authors provide a review on the most important mathematical models and algorithms developed for the solution of the BPP. In the experimental part they randomly generate 3840 instances with different number of items (50, 100, 200, 300, 400, 500, 750, 1000) and  bins capacities (50, 75, 100, 120, 150, 200, 300, 400, 500, 750, 1000). We use these instances because they are available online as well as their optimal solutions. For each instance of BPP, we generate the instance for 2-BCPP in the following way. At first, given an optimal solution of BPP, the bins (cells) are sorted in non-decreasing order quantities of items in the bin. Let $N$ be the number of the bins in optimal packing, and for each $i = 1, \dots, N$, let $n_i$ be the number of items in bin $i$. Then the 2-BCs for 2-BCPP are generated as follows. $n_1$ items in the first and the second bins create $n_1$ 2-BCs, which form 2-unions. To form one 2-BC, we take any item from the first bin and any item from the second bin. $n_2 - n_1$ items from the second and the third bins form the next 2-BCs in the same manner. The remaining $n_3 - (n_2-n_1)$ items of the third bin create 2-BCs with items from the fourth bin, and so on. The last bin $n_N$ may contain items that were not used in the constructed 2-BCs. Such items are removed, and the optimum for 2-BCPP becomes equal to $N-1$. Since the size of any bin and item in any BPP's instance may be arbitrary, the height of each bar in 2-BCPP is divided by the bin's capacity. As a result, for each $n \in \{25, 50, 100, 150, 200, 250, 375, 500\}$, we generate 480 different test instances with known optimal solution.

Let \emph{ratio} be $obj(X)/OPT$, where $OPT$ is the optimum and $obj(X)$ is the objective's value found by the algorithm $X\in\{M_w, A1, A2, GA\_LO\}$. Fig. \ref{fig7} presents the average and standard deviation values of the ratio for the BPP benchmarks-based instances. Again, as in the randomly generated instances, the algorithm $GA\_LO$ outperforms all other algorithms in accuracy and running time. Average ratio equals 1.034 for almost every value of $n$, except $25$ and $500$, where $R_{av}$ equals 1.037 and 1.035, correspondingly. Algorithms $M_w$ and $A2$ build similar solutions. For example, when $n=500$, the ratio is 1.11 for both of them. Algorithm $M_w$ in 34\% of cases performs only one iteration (builds one matching). $A1$ turns out to be the worst. For example, when $n=500$ its average ratio is 1.371, while average ratio for $M1_w$ equals 1.181.

Additionally, we analyze the absolute errors. Table \ref{table:4} presents the difference of the objective's values yielded by the approximation algorithms and optimum. In the majority of cases, the average absolute error of $GA\_LO$ appeared to be the smallest. For example, for $n = 500$, the average absolute errors of the approximation algorithms $M_w$, $M1_w$, $A1$, $A2$, and $GA\_LO$ equal 54.7, 80.1, 181.8, 52.8 and 15.6, correspondingly. Incidentally, for $n \in \{25, 50, 100$\}, $GA\_LO$ builds the optimal solutions in 40\%, 26.3\% and 14.8\% of cases, correspondingly. However, when $n \geq 150$, the algorithm $M_w$ builds an optimal solution more frequently than other algorithms. For example, when $n=500$, $M_w$ builds the optimal solution in 12.5\% of cases versus 10.6\% by $GA\_LO$ and 11.3\% by $A2$.

Consequently, we can conclude that all considered approximation algorithms solve the test instances obtained by the known BPP instances well enough. In many cases, optimal solutions were built. In the rest of the cases, the algorithms yield near-optimal solutions with a ratio close to 1. Like in the previous subsection, we have to state that algorithm $GA\_LO$ turns out to be the most beneficial.

\begin{figure}[t]
\includegraphics[width=\textwidth]{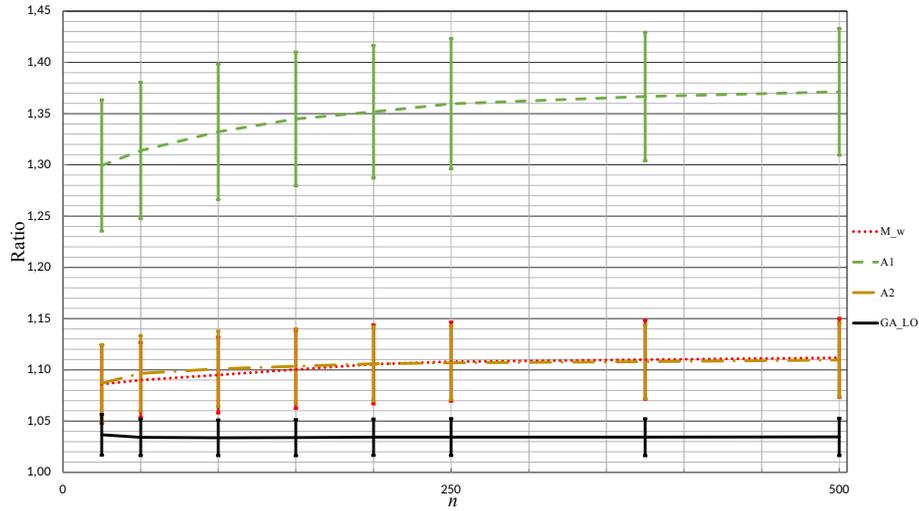}
\caption{Comparison of the algorithms on the BPP benchmarks based instances.} \label{fig7}
\end{figure}

\begin{table}[h!]
\centering
\setlength{\tabcolsep}{0.6pt}
{
\begin{tabular}{||c||c|c|c|c|c|c|c|c|c|c||}
\hline
    \multirow{2}{*}{$n$} & \multicolumn{2}{|c|}{$M_w$} & \multicolumn{2}{|c|}{$M1_w$} & \multicolumn{2}{|c|}{$A1$}
     & \multicolumn{2}{|c|}{$A2$} & \multicolumn{2}{|c|}{$GA\_LO$} \\
\cline{2-11}
 &  $max$& $av$ &  $max$& $av$ & $max$& $av$ & $max$& $av$ & $max$& $av$  \\

\hline

25   &7 &1.9   &11 &3.7  &12 &7.0  &6 &1.9   &5 &{\textbf{0.8}} \\
50   &11 &4.1   &24 &9.0  &24 &15.0  &12 &4.7   &6 &{\textbf{1.4}} \\
100  &26 &9.2   &49 &16.2  &48 &33.7  &23 &9.6   &14 &{\textbf{2.9}} \\
150  &39 &15.2   &74 &24.2  &72 &52.2  &33 &14.5   &23 &{\textbf{4.4}} \\
200  &54 &22.1   &99 &34.1  &98 &68.6  &43 &20.1   &27 &{\textbf{6.2}} \\
250  &62 &26.4   &124 &38.4  &118 &90.3  &63 &24.1   &39 &{\textbf{7.5}} \\
375  &91 &39.6   &183 &58.2  &177 &138.2  &79 &37.3   &55 &{\textbf{10.7}} \\
500  &120 &54.7   &249 &80.1  &235 &181.8  &103 &52.8   &83 &{\textbf{15.6}} \\

\hline

\hline
\end{tabular}
\caption{Upper bounds on the absolute errors of the algorithms on the BPP benchmarks based instances.
\label{table:4}
}}
\end{table}

\section{Conclusion}
We have tested several approximation algorithms with known guaranteed estimates to solve the strongly NP-hard problem of packing two-bar charts into the strip of minimal length. This problem is new for the optimization community. It is a generalization of the Bin Packing Problem, Strip Packing Problem, and 2-Dimensional Vector Packing Problem and a special case of the resource-constrained project scheduling problem where the jobs consume one renewable resource. For the earlier considered algorithms, we have found the guaranteed accuracy estimates. This paper performs a simulation using the randomly generated instances of different dimensions to compare the approximation algorithms and CPLEX for BLP and instances with known optimum. The numerical experiment shows the high efficiency of the greedy algorithm with the preliminary lexicographic ordering of the 2-BCs ($GA\_LO$) proposed firstly in \cite{Erzin20_2}, which significantly outperforms other algorithms in accuracy and runtime.

In this paper, we limited ourselves to approximation algorithms with guaranteed accuracy estimates. Furthermore, we deliberately estimate how much the accuracy of the considered algorithms is higher on average than the guaranteed accuracy. In future research, we are planning to test various other heuristics and metaheuristics, for which a priori accuracy estimates are not necessarily known.

\end{document}